\documentclass[a4paper, 11pt]{article}
\usepackage{comment} 
\usepackage{lipsum} 
\usepackage{fullpage} 
\renewcommand{\thesubsection}{\arabic{section}.\arabic{subsection}}
\renewcommand{\thesection}{\arabic{section}}
\usepackage{booktabs}
\usepackage{setspace}
\usepackage[hidelinks]{hyperref}
\usepackage{booktabs}
\usepackage[english]{babel}
\usepackage[utf8x]{inputenc}
\usepackage[margin=1in]{geometry}
\usepackage{graphicx}
\usepackage{subcaption}
\graphicspath{{desktop/}}
\usepackage[colorinlistoftodos]{todonotes}
\usepackage{amsmath}
\usepackage{amssymb}
\usepackage{mathrsfs}
\usepackage[colorinlistoftodos]{todonotes}
\usepackage[export]{adjustbox}
\usepackage{setspace}
\usepackage{natbib}
\setcitestyle{aysep={,}} 
\usepackage{pdflscape}
\usepackage{booktabs,caption}
\usepackage[flushleft]{threeparttable}
\usepackage{afterpage}
\usepackage{amsthm}
\usepackage{bm}

\newtheoremstyle{conditionstyle}
  {\topsep}   
  {\topsep}   
  {\normalfont}  
  {0pt}       
  {\bfseries} 
  {:}         
  {5pt plus 1pt minus 1pt} 
  {}          

\newtheoremstyle{theoremstyle}
  {\topsep}   
  {\topsep}   
  {\itshape}  
  {0pt}       
  {\bfseries} 
  {:}         
  {5pt plus 1pt minus 1pt} 
  {}          

\theoremstyle{conditionstyle}
\newtheorem{condition}{Condition}

\theoremstyle{theoremstyle}
\newtheorem{theorem}{Theorem}

\usepackage[toc,page]{appendix}

\newcommand\independent{\protect\mathpalette{\protect\independenT}{\perp}}
\def\independenT#1#2{\mathrel{\rlap{$#1#2$}\mkern3mu{#1#2}}}

\title{
Correcting Nonresponse Bias Using Panel Data \\
on Data Requests and Responses%
\footnote{This work was completed with funding from the National Science Foundation, Award \#2100017, The Wisconsin Alumni Research Foundation, The Richard M. Schulze Family Foundation, and American Family Insurance.
We thank Carolina Arteaga, Briana Ballis, Elizabeth Cascio, Soledad Giardili, Jonathan Guryan, Kevin Mumford, Peter Nenka, Victoria Prowse, Diego Salazar, Chris Taber, Justin Tobias, and Alex Torgovitsky as well as seminar participants at the Southern Economic Association Annual Meeting, the Society of Labor Economists Annual Meeting, and the Midwest Econometrics Group Conference for helpful comments.
We thank Rachel Keranen for copyediting assistance.
This research was conducted with administrative data provided by the University of Wisconsin-Madison.}
}

\author{
Clint Harris%
\footnote{
University of Wisconsin-Madison,
clint.harris@wisc.edu
}
\and
Jonathan T. Eckhardt%
\footnote{
University of Wisconsin-Madison,
jon.eckhardt@wisc.edu
}
\and
Brent Goldfarb%
\footnote{
University of Maryland,
brentg@umd.edu
}
}

\date{\today}

\begin{document}
\maketitle


\begin{abstract}
When subjects who respond to requests for data, 
such as in surveys or post-treatment follow-up, 
are unrepresentative of the population of interest, inferences drawn from the data 
can be misleading.
We show that if
subjects' accumulated requests and responses over time are recorded and organized as panel data, 
requests can be used 
as instruments 
in standard nonresponse corrections.
Randomizing total requests between subjects 
is insufficient 
for identification 
when responses depend on time
and is 
unnecessary 
when they do not.
We demonstrate our method by estimating an 18-percentage-point gender gap in entrepreneurial career intentions using a survey of undergraduates at the University of Wisconsin-Madison.

\hfill

\end{abstract}

\clearpage

\section{Introduction}
Researchers, governments, and private companies routinely request voluntary data submission, often more than once, from samples of populations of interest.
These data gathering efforts usually only succeed with a portion of subjects.
For example, this pattern occurs in experiments with delayed measurement of outcomes, such as those investigating effects of masks on COVID 19 transmission \citep{abaluck2022impact}, effects of education on STIs and fertility \citep{duflo2015education}, and effects of neighborhoods on health and safety \citep{katz2001moving}.
The same pattern arises in survey data gathered 
by universities requesting student course reviews, 
by political pollsters seeking to predict election results, 
by private businesses seeking feedback on product satisfaction, and 
by governments seeking information from individuals, businesses, and nonprofit institutions.%
\footnote{For instance, the Census Bureau \citep{spiers2022strategic}, the Bureau of Labor Statistics \citep{nlsy97retention}, and the Internal Revenue Service \citep{irscp59} all document procedures for making multiple attempts to elicit responses from surveyed individuals for data they collect.}
If the individuals who voluntarily provide data are not representative of the population as a whole, inferences drawn from the data about the real world can be misleading.

An insight from \cite{heckman1976common,heckman1979sample} is that variation in data observation rates between otherwise similar subpopulations can be used to correct for selection bias.
To perform such a correction, it is necessary to retain data on nonrespondents along with respondents.
Our core insight is that response behavior is observed for each subject each time they are sent a data request, rather than once.
Retaining all such observations, at odds with standard practice, naturally provides variation in data observation rates between otherwise similar subpopulations.
Briefly, we propose that researchers retain panel data on subjects' responses at the time of each data request so that cumulative request counts can be used as instrumental variables to correct for nonresponse bias.%
\footnote{
Methods that leverage instruments include parametric methods such as those of \cite{heckman1979sample} and the extension to binary outcomes of \cite{van1981demand}, as well as nonparametric bounding methods such as those of \cite{manski1990nonparametric}, \cite{manski1994selection}, and \cite{manski2000monotone}.
}

We join \cite{behaghel2015please} and \cite{dutz2021selection} in inferring subjects' responsiveness from their response timings, and using this information to correct for nonresponse bias.
Our primary distinction from their methods is to use response and request timing to construct a panel dataset containing subjects' responses as of each data request, rather than using response timing to order responsiveness between subjects in cross-sectional data.
Constructing a panel dataset has two advantages.
First, it makes explicit the assumption that potential responses are unaffected by time, which is made implicitly when potential responses are modeled cross-sectionally.
Second, 
it converts response timing data into a format that accommodates existing nonresponse corrections that leverage instruments, as well as accommodating additional instruments that may or may not vary over time, such as randomized response incentives, without requiring revisions to estimation procedures.

Our main result is a proof that local average responses (LARs) for individuals who comply with each data request are identified from  
the joint distribution of responses and response rates for all requests.
Identification of multiple local average responses is sufficient to estimate or bound population averages of requested variables using existing estimation methods that correct for selection bias, which are united by the marginal treatment response framework of \cite{heckman2005structural,heckman2007econometric}.
We establish identification under assumptions that resemble those of \cite{imbens1994identification} and \cite{angrist1996identification}.
Identification follows regardless of whether requests are randomized or sent uniformly to all subjects if potential responses are time-invariant, and identification fails regardless of randomization if potential responses vary systematically with time.

In Section \ref{sect:application}, we demonstrate the use of nonrandom request instruments by estimating the gender gap in entrepreneurial career intentions among undergraduates using survey data from the University of Wisconsin-Madison.
Estimating a parametric \cite{van1981demand} model for binary outcomes, we find an 18-percentage-point nonresponse-corrected male-female gap in entrepreneurial intentions. 
We find statistically insignificant evidence of positive selection for men and negative selection for women driving the 20 percentage point uncorrected gap among respondents.
This model is identified with two requests, so we also implement and pass a ``predictable trends'' overidentification test of the null that later requests have no direct effect on the outcome, supporting the validity of our parametric and nonparametric assumptions in this application.
In Supplemental Appendix \ref{appendix:selection_in_surveys}, we estimate labor market outcomes using a synthetic version of the Norway in Corona Times survey constructed to match estimates reported by \cite{dutz2021selection}.
The 95\% confidence intervals of our estimated population means cover ground truth averages from administrative data reported by these authors for five out of six survey variables, with respondent averages differing significantly from the ground truth for all variables.

\section{Correcting Nonresponse With Multiple Requests}\label{sect:Theory}
We consider a random sample of $N$ subjects indexed by $i$ who are 
each
monitored regarding their responses to data requests in time periods indexed by $t=0,1,...,T$.
Subject $i$ receives data requests for an unobserved time-invariant \textit{requested variable}, denoted by $Y_{i}^*$, with their observed \textit{requests} accumulated by the end of time $t$ given by $R_{it}$. 
In each time period, 
we also observe the \textit{response choice} $S_{it}$, with $S_{it}=1$ if $i$ responds in time $t$ and $S_{it}=0$ if they do not, 
and the \textit{response}, $Y_{it}$, observed if and only if $S_{it}=1$.
We define the time index
such that $S_{i0}=0$ and $R_{i0}=0$ for all $i$
and
$R_{it}-R_{i,t-1}
\in \{0,1\}$ for all $i$ and $t$.\footnote{
For any realistic application in which it would seem natural to organize a dataset such that $S_{it}=1$, $R_{i0}>0$, or $R_{it}>1+R_{i,t-1}$,
such as those involving unsolicited responses,
lost information regarding early request timing,
or
rapid accrual of requests,
an alternative dataset 
recording identical request and response behavior
can be constructed 
with a more granularly defined time index
where these constraints hold.
}
We assume
that
$S_{it}=1$ at most once for each $i$,
with subject $i$'s \textit{retained response choice} in time $t$ 
denoted by $\hat S_{it} = \sum_{t'=0}^t S_{it'}$,
and their \textit{retained response} in time $t$ 
denoted by $\hat Y_{it} = \sum_{t'=0}^t S_{it'}Y_{it'}$.

In settings with nonresponse, the sample mean of retained responses in period $T$ among respondents is a consistent estimate of $\mathbb{E}[\hat Y_{iT}|\hat S_{iT}=1]$, with expected bias relative to $\mathbb{E}[Y_{i}^*]$ given by
\begin{equation*}\label{Bias}
B_w = \mathbb{E}[\hat Y_{iT}|\hat S_{iT}=1] - \mathbb{E}[Y_{i}^*].
\end{equation*}
With random nonresponse and no measurement error, $\mathbb{E}[\hat Y_{iT}|\hat S_{iT}=1] = \mathbb{E}[Y_{i}^*]$, 
and 
population moments are identified from sample data.
If these conditions are not met, identification of population averages requires additional information.

\subsection{Modeling Nonresponse}
We assume that observed variables are realizations of potential outcomes, following \cite{rubin1974estimating}.
First, let $S_{it} 
=  S_{it}(R_{it})
        \prod_{t'=0}^{t-1}
        (1- S_{it'}(R_{it'}))$
where $S_{it}(r) \in \{0,1\}$ indicates subject $i$'s \textit{willingness to respond} to $r$ requests at time $t$.%
\footnote{
We use this definition in lieu of $S_{it} 
=  S_{it}(R_{it})$ because we assume that subjects respond at most once.
}
Similarly, let $Y_{it} = Y_{it}(S_{it},R_{it}) = 
Y_{it}(1,R_{it})S_{it}
+
Y_{it}(0,R_{it})(1-S_{it})$
with $Y_{it}(1,r)$ denoting the response subject $i$ would give at time $t$ if they were to respond given $r$ accumulated requests, with $Y_{it}(0,r)$ missing for all $i$, $t$, and $r$.
$S_{it}(r)$ and $Y_{it}(1,r)$ are defined for all $r \in \{0,...,\bar r\}$,
while they are observed only for the realized $R_{it}$.

Following the marginal treatment response framework of \cite{heckman2005structural,heckman2007econometric}, 
we assume that selection bias is driven by dependence between responses and response aversion.
Specifically, we define the marginal survey response function $m(u) \equiv \mathbb{E}[Y_i^* | U_{it}=u]$ where $U_{it} \in [0,1]$ denotes the \textit{response aversion} percentile of subject $i$ at time $t$.
We assume that response aversion determines potential response choices such that
\hfill
\begin{equation}\label{assumption:vytlacil_2002}
    S_{it}(r) = 1(U_{it} \leq P(r)),
\end{equation}
for all $r$, 
where $P(r) = \mathbb{E}[S_{it}(r)]$ gives the \textit{willing response propensity} for request $r$.
We then have
\hfill
\begin{equation}\label{equation:LAR_to_MSR}
    \mathbb{E}[Y_i^*|P(r')<U_{it} \leq P(r)]
    =
    \mathbb{E}[Y_i^*|S_{it}(r)-S_{it}(r')=1],
\end{equation}
for all $r$ and $r'$.

It follows that we can learn about $m(u)$ if we can estimate $\mathbb{E}[Y_i^*|S_{it}(r)-S_{it}(r')=1]$ for some pairs of $r$ and $r'$, which we refer to as the \textit{local average response} for $r-r'$ \textit{compliers}.
For instance, 
a single identified local average response is sufficient to identify $m(u)$ for all $u$ under the assumption that it is constant in $u$,
and
two identified local average responses are sufficient to test this assumption
based on its implication that they should be equal.
If $m(u)$ is parameterized, as in a \cite{heckman1979sample} model, point estimation is possible if there are weakly more identified local average responses than parameters.
Otherwise, 
bounds on $m(u)$ can be obtained 
using methods such as those of \cite{horowitz2000nonparametric}, \cite{manski2009more}, or \cite{lee2009training}.

In data retaining only observations with $t=1$,
the 
first request's
local average response 
is identified as
\hfill
\begin{equation}\label{eq:LAR_t1_identification}
    \mathbb{E}[Y_i^*|S_{i1}(1)-S_{i1}(0)=1] 
= 
\frac{
\mathbb{E}[Y_{i1}S_{i1}|R_{i1}=1]
-
\mathbb{E}[Y_{i1}S_{i1}|R_{i1}=0]
}
{
\mathbb{E}[S_{i1}|R_{i1}=1]
-
\mathbb{E}[S_{i1}|R_{i1}=0]
}
\end{equation}
if
$R_{i1}$ is a
valid instrument for response choices
under the assumptions 
of
\cite{imbens1994identification}.
Because time mechanically elapses between requests for all individuals, 
local averages responses for requests beyond $R_{it}=1$ are unidentified in data with $t=1$.
Our primary contribution is
to establish that local average responses
are identified 
in data retaining any values of $t$
as
\hfill
\begin{equation}\label{eq:LAR_identification}
    \mathbb{E}[Y_i^*|S_{it}(r)-S_{it}(r')=1] 
= 
\frac{
\mathbb{E}[\hat Y_{it}|R_{it}=r]
-
\mathbb{E}[\hat Y_{it}|R_{it}=r']
}
{
\mathbb{E}[\hat S_{it}|R_{it}=r]
-
\mathbb{E}[\hat S_{it}|R_{it}=r']
}
\end{equation}
for all $r$ and $r'$ in the support of $R_{it}$,
under assumptions 
implied by those of \cite{imbens1994identification} 
for the partition defined by $t=1$.
When other partitions are chosen,
such as
the full panel $t=\{0,1,\dots,T\}$ 
or
the partition $t=T$
implicit in
cross-sectional data containing multiple requests,
our assumptions impose 
strong but intuitive
behavioral restrictions 
on 
how requests and responses evolve over time,
enabling identification of local average responses for more values of $r$ and $r'$ in applications where multiple requests accrue over time.


\subsection{Local Average Response Identification Assumptions}
Our goal is to attribute differences in $\hat Y_{it}$ across different values of $R_{it}$ to averages of $Y_i^*$ for compliers who require more or fewer requests to induce response.
If more requests are systematically sent to subjects or time periods with unrepresentative potential responses, we cannot reliably attribute differences in responses between requests to marginal compliers.
We begin by assuming that requests are sent in such a way that subjects' potential responses and potential response choices across time are unrelated to their request histories.



\begin{condition}[Intertemporal Independence]\label{cond:independence} 
\hfill

$\left(Y_{it}(1,r),\, S_{it}(r)\right) 
\independent \left(R_{i,t-1},R_{it},\dots,R_{iT}\right)$ for all $r$. 


\end{condition}
Our independence condition strengthens that of \cite{imbens1994identification} 
by ruling out endogenous data requests sent to subjects 
with 
unrepresentative
imminent,
present,
or past
potential responses.
Independence of $Y_{it}(r),S_{it}(r)$ from only $R_{it}$ 
allows for cases where potential responses revealed when $R_{it}>R_{i,t-1}$,
and retained as $\hat Y_{it'}$ for $t'>t$
are unrepresentative
of potential outcomes for observations with the same values of $R_{it}$,
whereas our stronger assumption rules this out.
When our stronger assumption is violated 
and 
within-row independence is satisfied,
even surveys with 100\% response rates 
for the first request 
can obtain arbitrarily unrepresentative responses.%
\footnote{
To illustrate,
consider a case in which 
$t$ indexes days,
request $R_{it}=1$ arrives on a weekend for all subjects,
all subjects respond immediately,
and 
observations are retained such that $T=7$.
Row-wise independence between coinciding potential outcomes and accrued requests 
can be satisfied even with 
arbitrarily unrepresentative potential responses on weekends.
Meanwhile, 
our intertemporal independence 
rules this out by
requiring potential responses on days with different future request schedules to be drawn from the same distribution.
}

Randomizing 
total 
requests
between subjects,
as proposed by \cite{dinardo2021practical}, 
ensures 
that accrued requests at time $T$ are 
unrelated to 
potential responses
for each time period.
However,
subjects select the time period in which they respond,
and
retained responses at time $T$ have been carried forward from prior periods.
It follows that 
\textit{retained responses}
at time $T$ 
are 
only 
representative of 
\textit{potential responses} at time $T$
if 
subjects 
do not systematically respond in periods $t' \leq T$ in which
they have unrepresentative potential responses.
Our stronger condition
imposes this restriction for all $t'\leq t$, for all $t$,
rather than only for $t=T$.
At the cost of 
ruling out 
cases in which 
unrepresentative responses 
cancel out between all time periods 
but
not within each time period,
our assumption
enables identification of more local average responses
in cases where the support of $R_{it}$ exceeds that of $R_{iT}$,
regardless of request randomization.

We next assume that requests do not directly affect potential responses, allowing us to write $Y_{it}(1) = Y_{it}(1,r)$.
\begin{condition}[Exclusion]\label{cond:exclusion} 
\hfill

$Y_{it}(1,r) = Y_{it}(1,r')$ for all $r$ and $r'$, for all $i$ and $t$.
\end{condition}
\noindent The importance of avoiding experimenter demand effects or priming in surveys is well established \citep{stantcheva2023run}.
Our assumption expresses this as a special case of similar assumptions made in the context of sample selection by \cite{heckman1979sample} and in the context of treatment effect identification by \cite{angrist1996identification}.

With independence and exclusion, differences in responses to different requests
can only be explained by differences in responses between distinct complier groups, 
defined by the values of $r$ for which $S_{it}(r)=1$.
Our next assumption restricts the possible number of these groups by imposing a natural ordering on them,
and imposes stability across time in this ordering.
\begin{condition}[Intertemporal Monotonicity]\label{cond:monotonicity} \hfill

a. $ S_{it}(r) \geq  S_{it}(r')$ for all $r>r'$, for all $i$ and $t$.

b. $ S_{it}(r) =  S_{it'}(r)$ for all $t$ and $t'$, for all $i$ and $r$.


\end{condition}
\noindent Our monotonicity condition adopts the positive sign restriction of \cite{angrist1996identification} presented in the context of intent-to-treat instruments,
following the intuition that it is similarly implausible that ``intent-to-respond'' survey request instruments would reduce average response rates.
With the additional time-invariance assumption on potential response choices,
our monotonicity condition 
implies that $\hat S_{it} = S_{i}(R_{it})$, which enables identification of $P(r)$ and $\mathbb E[Y_{it}(1,R_{it})|S_i(R_{it})=1]$ from average retained response choices and retained responses.

Relatedly, we assume that subsequent requests strictly increase response rates, 
slightly strengthening the first part of our monotonicity assumption.
\begin{condition}[Relevance]\label{cond:relevance}
$P(r)>P(r')$ for all $r>r'$.
\end{condition}
\noindent As in treatment effects applications, relevance can be tested by estimating $P(r)$ for all $r$, with $P(r)=\mathbb{E}[\hat S_{it}|R_{it}=r]$ if intertemporal independence and intertemporal monotonicity hold.

Independence, exclusion, and monotonicity allow us to identify 
self-reported
local average responses for compliers to each request.
We also assume that each complier group gives accurate responses on average.
\begin{condition}[Accurate Aggregate Compliance]\label{cond:classical_measurement_error} \hfill

a. $\mathbb{E}[Y_{it}(1,r)-Y_i^*]=0$ for all $r$.

b. $Y_{it}(1,r)-Y_i^* \independent S_{it}(r)$ for all $r$.

\end{condition}
\noindent This assumption is weaker than $Y_{it}(1,r) = Y_i^*$, which is common in the related literature, because it allows for idiosyncratic shocks that cancel out within complier groups, such as if students' potential responses regarding career intentions change over time in response to new information, such as grades on assignments.
Independence rules out differential 
selection both between subjects and between response times,
while exclusion rules out measurement error caused by requests,
but neither rules out systematic measurement error for entire complier groups.
For example, 
if
late request compliers respond with low effort,
for instance by drawing responses uniformly from the support of $Y_i^*$,
independence and exclusion can be satisfied
despite 
systematic inaccuracy in responses.
If this assumption is violated 
and our other assumptions hold,
selection corrections using request instruments
can provide
representative
but potentially inaccurate
information on potential responses
by addressing selection bias without addressing measurement error.

\subsection{Local Average Response Identification Proof}\label{subsection:proof}
Our main result proves that local average responses are identified under conditions 1–5.
\begin{theorem}
    If Conditions 1–5 hold, the local average response for subjects who respond to request $r$ but not $r'$ is identified as
        \begin{align*}
        \mathbb{E}[Y_i^*| S_i(r)- S_i(r')=1] = 
        \frac{
        \mathbb{E}[\hat Y_{it}|R_{it}=r] 
- 
\mathbb{E}[\hat Y_{it}|R_{it}=r']
        }
        {
        \mathbb{E}[\hat S_{it}|R_{it}=r]-\mathbb{E}[\hat S_{it}|R_{it}=r']
        }
    \end{align*}
    from the joint distribution of $Y$, $S$, and $R$ for any $r>r'$ such that $\mathbb{E}[\hat Y_{it}|R_{it}=r]$ and $\mathbb{E}[\hat Y_{it}|R_{it}=r']$ are finite and $P(r)>P(r')$.
\end{theorem}
\begin{proof}
    Assuming that Conditions 1–3 hold,
        we have
    \begin{align*}
        \mathbb{E}&[\hat Y_{it}|R_{it}
        =r]
        =
        \mathbb{E}
        \left [
        \sum_{t'=0}^t
        Y_{it'}(1,R_{it'})
        S_{it'}(R_{it'})
        \prod_{t''=0}^{t'-1}
        (1- S_{it''}(R_{it''}))
        |R_{it}=r
        \right ]
        \\
        &=
        \mathbb{E}
        \left [
        \sum_{t'=1}^t
        Y_{it'}(1,R_{it'})
        (
         S_{i}(R_{it'})
        -S_{i}(R_{i,t'-1})
        )
        |R_{it}=r
        \right ]
        \\
        &=
        \sum_{t'=1}^t
        \mathbb{E}
        \left [
        \sum_{r'=0}^{ r}
        Y_{it'}(1,r')
        (
         S_{i}(r')
        -S_{i}(r'-1)
        )
        (
        R_{it'}-R_{i,t'-1}
        )
        I(R_{it'}=r')
        |R_{it}=r
        \right ]
        \\
        &=
        \sum_{r'=1}^{ r}
        \mathbb{E}
        \left [
        Y_{it}(1)
        (
         S_{i}(r')
        -S_{i}(r'-1)
        \right ]
        \sum_{t'=1}^t
        Pr(
        R_{it'}=r',
        R_{it'}-R_{i,t'-1}=1
        |R_{it}=r)
        \\
        &=
        \sum_{r'=1}^{ r}
        \mathbb{E}
        \left [
        Y_{it}(1)
        (
         S_{i}(r')
        -S_{i}(r'-1)
        \right ]
        \\
        &=
        \mathbb{E}[
        Y_{it}(1)S_{i}(r)
        ]
        . 
    \end{align*}
The first equality rewrites retained responses as realizations of potential responses.
The second equality follows from monotonicity and the mechanical constraint $R_{it} \geq R_{i,t-1}$,
which require that willingness to respond is an absorbing state over time periods.
The third equality follows from 
monotonicity and
the constraint 
$R_{it}-R_{i,t-1}
\in \{0,1\}$,
under which 
subject $i$'s response is observed
only
if the survey administrator sends $i$ a request at time $t$
and
$i$ responds to that request.
The fourth equality
follows from
independence and exclusion,
with the substitution
of $Y_{it}(1,r')S_i(r')$ for $Y_{it'}(1,r')S_i(r')$ inside the expectation
following from independence because the request schedule $\bm{R}_{it'} = (R_{i,t'-1},R_{it'},\dots,R_{iT})$
necessarily differs in length
from 
$\bm{R}_{it''}$ for all $t' \neq t''$, for all $i$.
The fifth equality follows from 
the constraint $R_{it}-R_{i,t-1}
\in \{0,1\}$,
which ensures that 
$\sum_{t'=1}^t
        Pr(
        R_{it'}=r',
        R_{it'}-R_{i,t-1}=1
        |R_{it}=r) = 1$ for all $r' \leq r$.
The sixth equality follows from monotonicity and the restriction $S_{i0}=R_{i0}=0$,
which imply that $S_i(0)=0$.

Assuming that conditions 4–5 also hold, 
the instrumental variables estimand identifies the local average response for $r>r'$ as
\begin{align*}
\frac{
        \mathbb{E}[\hat Y_{it}|R_{it}=r]
        -
        \mathbb{E}[\hat Y_{it'}|R_{it'}=r']
        }
        {
        \mathbb{E}[\hat S_{it}|R_{it}=r]-\mathbb{E}[\hat S_{it'}|R_{it'}=r']
        }
        &=
        \frac{
        \mathbb{E}[
        Y_{it}(1)Si(r)
        ]
        -
        \mathbb{E}[
        Y_{it}(1)Si(r')
        ]
        }
        {
        P(r)-P(r') 
        }
        \\
        &=
        \frac{
        \mathbb{E}[
        Y_{i}^*Si(r)
        ]
        -
        \mathbb{E}[
        Y_{i}^*Si(r')
        ]
        }
        {
        P(r)-P(r') 
        }
        \\
        &=
        \mathbb{E}[
        Y_{i}^*
        |
         S_{i}(r)
        - S_i(r')=1
        ],
        \end{align*}
        where the first equality makes 
        the substitution $\mathbb{E}[\hat Y_{it}|R_{it}=r] = 
        \mathbb{E}[Y_{it}(1)
         S_{i}(r)
        |R_{it}=r]$ as derived above 
        and 
        the substitution 
        $\mathbb{E}[\hat S_{it}|R_{it}=r] = 
        P(r)
        $
        by following the same steps for $\hat S_{it}$. 
        The second equality follows from accurate aggregate compliance.
        The third equality follows from monotonicity and relevance.
\end{proof}

Our result establishes conditions under which data requests are valid instruments for retained responses in data where subjects' requests and responses are observed multiple times.
It follows that established estimators that use instrumental variables 
to address selection bias, such as that of \cite{heckman1979sample},
can be used in settings without explicitly randomized response incentives.
The insight underlying these methods is that differences in responses between otherwise-similar subjects with different response rates can be explained by dependence between potential responses and response aversion. 
A more thorough review of such estimators is provided by \cite{dutz2021selection}, along with a generalized model of survey aversion with heterogeneous effects of instruments on  responsiveness.

\subsection{Practical Considerations Regarding Assumption Sensitivity}
The sensitivity of empirical estimates to between-subjects violations of our assumptions
has been thoroughly discussed in the literature on instrumental variables,
so we focus on intertemporal violations and compare them to their between-subjects counterparts.
Our identification assumptions closely resemble those of a before-after comparison design,
which is 
emphasized by 
the case 
in which $R_{it}=t$ for all $i$
and
is
obscured, 
but not generally resolved,
by
the case in which
observations are only retained if $t=T$.
In simple terms,
interpreting
survey data at face value
requires 
that 
request accrual
and
response timing are ignorable
regardless of 
subjects' response rates 
or
researchers' time indexing choices.
Independence violations
due to the passage of time
are the most salient threat to identification in before-after designs,
while the assumption 
$S_{it}(r)=S_i(r)$ is likely to be violated when subjects delay their responses,
so we focus on those in this section.%

To facilitate the discussion,
let $X_i$ denote a set of subject characteristics that are fixed over time,
and 
let $\tau_{it}$ denote time-varying information subject $i$  at time $t$.
$X_i$ may contain elements such as gender,
while $\tau_{it}$ may contain
idiosyncratic elements,
such as political or economic events that occur in the interval defining time $t$ for subject $i$,
cyclical elements,
such as indicators for the weekdays or hours of the day at which subject $i$ is available to respond in time interval $t$,
and 
secular elements,
such as subject $i$'s age at time $t$.%
\footnote{
Independence violations from idiosyncratic and cyclical time shocks fall under \citeauthor{campbell1979quasi}'s \citeyearpar{campbell1979quasi} ``history'' category 
in their list of threats to internal validity in before-after designs,
with secular trends falling under ``maturation''.
Exclusion and monotonicity violations 
relate to their 
``testing'' and ``mortality'' threats to validity,
respectively.
}
We organize our discussion of practical considerations regarding assumption sensitivity
with respect to the choices of each decision maker in the data-generating process.

\subsubsection{Sensitivity to Subject Behavior}
It is intuitive that all of our assumptions are satisfied for sufficiently well-behaved subjects regardless of 
data collection choices 
and 
data analysis choices.
In terms of the model,
arbitrarily well-behaved subjects have
$Y_{it}(1,r)S_{it}(1) = Y_i^*$ for all $i$, $r$, and $t$
regardless of
the survey instrument chosen by the data collector,
the times at which the data collector sends requests,
the times at which the subject responds to the request,
and
the time index definition chosen by the data analyst.
The only scope for assumption violations in such a case involve researcher misrepresentation of the subject population of interest,
with an independence violation made apparent by 
augmenting the sample with additional subjects 
as needed to produce a representative sample
and setting $R_{it}=0$ for all $t$ for such subjects.

\subsubsection{Sensitivity to Data Collection Design}
Independence is violated if potential responses 
vary with time
and responses are systematically obtained at unrepresentative times,
with dependence 
between $(R_{i,t-1},\dots,R_{iT})$ and $\tau_{it}$ 
presenting the same problems as
dependence
between $R_{iT}$ and $X_i$ for the seemingly cross-sectional case.
Secular time dependence can be mitigated by careful wording in survey instruments to emphasize time-invariance for requested variables.
For instance, requesting 
that subjects' report their earnings for a particular calendar year, 
which has a time-invariant correct answer,
is likely preferable to requesting
that they report
their yearly salary, 
which can change between requests.

Cyclical time dependence,
such as cases in which subjects systematically misreport 
in opposite directions
on weekends/weekdays or mornings/evenings,
can be mitigated by randomizing subjects' response times,
or by 
referencing subject matter expertise to restrict requests to times in which subjects are likely to give representative responses.%
\footnote{
For subjects who respond promptly to data requests,
response timing can implicitly be randomized 
by randomizing request timing.
For subjects who selectively choose response times,
the data collector can 
randomize response timing by 
between-subjects randomization of response time encouragements,
such as by 
including a proposed response time in each data request,
providing a monetary incentive for particular response times,
and/or
enforcing randomized response timing by randomizing subjects' access to the survey across time,
akin to a \cite{zelen1979new} single-consent design.
}
Finally, 
endogenous followup survey designs,
such as those following the recommendations of \cite{groves2006responsive},
violate unconditional independence.
However,
if survey administrators
retain and report observed covariates that determine followup, 
conditional independence will potentially hold conditional on the
covariates used to determine followup.

\subsubsection{Sensitivity to Panel Construction}
When intent-to-treat instruments are cross-sectionally randomized for some subpopulations and not others,
conventional practice employs sample restrictions or covariate adjustments
to address resulting independence violations.
The same reasoning applies for data collection efforts across time when representative subjects and times are surveyed in some time periods and not others.
In the case of survey followup, 
the most natural threat is selection on unobservables,
such as when additional requests are sent to low-response subpopulations who are observed by the survey administrator but not the data analyst.
Such situations can be addressed by restricting the analysis to observations with small values of $t$,
with the extreme case for the partition with $t=1$ shown in equation (\ref{eq:LAR_t1_identification}).
In cases where $R_{it} = R_{t}$ for all $i$ for $t\leq t^*$ for some $t^* \leq T$, endogenous followup to selected subjects is ruled out.
Restricting the sample to such observations thus serves as a useful robustness check in cases with the possibility of undocumented endogenous followup.

In cases where 
subject matter expertise 
suggests responses at particular times
may be unrepresentative for the population of interest,
for instance if some requests regarding a mental health survey coincide with holidays,
dropping observations with such $t$ may improve inferences.
When request timing is randomized between subjects,
statistical tests of equality between the same
local average responses 
in different time periods
can potentially detect unrepresentative time periods.

Finally,
cases in which many individual time periods are unrepresentative in unknown ways can be addressed via time index coarsening,
for instance,
as implemented
by keeping data for $t \in \{s,2s,\dots,T\}$,
for some whole number $s$.
In such cases,
we obtain
\begin{align*}
    \hat Y_{it} = 
        \sum_{r'=1}^{ r}
\sum_{t'=1}^t
        \mathbb{E}
        \left [
        Y_{it'}(1)
        (
         S_{i}(r')
        -S_{i}(r'-1)
        |
        R_{it'}=r',
        R_{it'}-R_{i,t'-1}=1,
        R_{it}=r
        \right ]
        P_{t'}(r',r),
\end{align*}
in line 4 of the proof,
where $P_{t'}(r',r) = Pr(
        R_{it'}=r',
        R_{it'}-R_{i,t-1}=1
        |R_{it}=r)$.
Knife edge cases in which case the equality
\begin{align*}
\sum_{r'=1}^{ r}
\sum_{t'=1}^t
        &\mathbb{E}
        \left [
        Y_{it'}(1)
        (
         S_{i}(r')
        -S_{i}(r'-1)
        |
        R_{it'}=r',
        R_{it'}-R_{i,t'-1}=1,
        R_{it}=r
        \right ]
        P_{t'}(r',r)
        \\
        = 
        \sum_{r'=1}^{ r}
        &\mathbb{E}
        \left [
        Y_{it}(1)
        (
         S_{i}(r')
        -S_{i}(r'-1)
        \right ],
\end{align*}
    holds 
    despite independence violations
    occur when unrepresentative responses obtained between retained time periods 
    exactly
    cancel out.
    While it may appear unpalatable to assume that independence violations cancel out between observations,
    such 
    an assumption is implicit for any time index,
    with researcher judgment used to pool responses obtained at different moments in continuous time
    into discrete time intervals.
    
    As a practical matter,
    dropping time periods 
    is beneficial 
    when
    it
    reduces overweighting of unrepresentative early responses,
    such as when subjects have idiosyncratic availability with $S_{it}(r) \neq S_{it'}(r)$,
    and 
    when 
    independence violations across more granular time intervals cancel out.%
    It is costly when
    it
    reduces the support of $R_{it}$ in the retained data.
    The extreme case of this is exactly the procedure of \cite{dinardo2021practical} in which only data for time $T$ is retained,
    with the related case in which $t$ is only retained if $R_{it}>R_{i,t-1}$ for some $i$
    corresponding closely to the reminder-leveraging procedures of \cite{behaghel2015please} and \cite{dutz2021selection}.
    Such procedures can be expected to work 
    well
    when independence violations occur due to cyclical temporal confounding within natural time intervals,
    and
    data collectors send requests
    such that independence violations cancel out within time intervals between requests.%
    \footnote{
    A natural example of independence violations canceling out over time involves $t$ being defined in hours, 
    with requests sent in the morning.
    Estimates obtained from an hourly panel in such a case would overweight morning respondents, 
    whereas those obtained from a daily panel would not.
    }

\section{Application: Gender Entrepreneurship Gap}\label{sect:application}
We use our method to estimate the gender gap in entrepreneurial aspirations among undergraduate students.
There is a large gender gap in entrepreneurship among working-age adults \citep{aldrich2005entrepreneurship}, perhaps partially driven by a gender gap in early-stage funding for new ventures \citep{canning2012women,greene2003women}.
Our paper contributes to this literature by investigating pre-market gender gaps in entrepreneurial career interest, which may be more likely to develop due to intrinsic interest gaps or pre-market discrimination rather than labor market discrimination.
This application showcases the importance of researcher judgments in defining the time index, adjusting for nonrandom variation in requests, and testing modeling assumptions in cases where there are more requests than are needed for identification.
In Supplemental Appendix \ref{appendix:selection_in_surveys}, 
we demonstrate our method's accuracy 
despite minimal data requirements
by showing that nonresponse-corrected estimates are closer to ground truth values from administrative data than 
respondent averages for the Norway in Corona Times Survey,
while relying only on
summary statistics of survey and ground truth measures 
taken from the publicly published tables of \cite{dutz2021selection}.

\subsection{Data}
We use a survey of the entrepreneurial intentions of the undergraduate population at University of Wisconsin–Madison that was implemented every fall from 2015 to 2022, as well as in spring 2020 and 2021. 
Students could choose between ``Yes,'' ``No,'' or ``I Don't Know,'' when responding to a question regarding whether they intend to pursue a career in entrepreneurship.
We code an ``Intention'' variable equal to 1 if the student answered ``Yes'' and 0 if they answered ``No'' or ``I Don't Know.''
We observe term-specific survey request stratum indicators for groups of students who are sent requests at the same time through survey software.
We observe timestamps for each request made to each stratum and for each student's response if they respond.
Students who opted out of the survey in a prior year have no stratum indicators and no request or response timestamps in the raw data, whereas students who respond to early requests have stratum indicators and request timestamps for all \textit{intended} requests in raw data, though they did not receive additional actual requests after responding.

Overall, 103,536 unique students and 333,201 total student-terms were surveyed.
We merge the survey data to administrative records containing information on students regardless of whether they responded to the survey.
Business Major and STEM Major are set to one if the student has any major in the term on the business school's list of majors or on the list of STEM majors provided by \cite{ice2020}, respectively.
GPA refers to the student's cumulative grade point average that varies across terms, ACT Math is the average of students' ACT math test scores on file, and ACT Verbal is the average of the sum of English and Reading portions of the ACT test.%
\footnote{We convert SAT scores into ACT scores for students without ACT scores using the 2018 ACT/SAT concordance tables provided by act.org.}
Female, Racial Minority, and International are binary indicators for female gender from binary gender information in raw data, listing a race other than white, and having a recorded country of origin other than the U.S., respectively. 
We drop 6,075 student-terms with missing GPAs (students who dropped out in their first term) and 33,296 observations with missing ACT scores, leaving us with 290,588 student-terms and 82,000 unique students.
There were 42,902 total survey responses, with a response rate of approximately 14\% for men and 16\% for women.
Sample means of key variables broken down by response timing are provided in Table \ref{tab:descriptives}.

\begin{table}[htbp]\centering
\def\sym#1{\ifmmode^{#1}\else\(^{#1}\)\fi}
\caption{Summary Statistics}
\footnotesize{
\label{tab:descriptives}
\begin{tabular}{@{\extracolsep{4pt}}l*{10}{c}}
\toprule
& \multicolumn{3}{c}{Early Respondents} & \multicolumn{3}{c}{Late Respondents} & \multicolumn{3}{c}{Nonrespondents} \\
\cline{2-4} \cline{5-7} \cline{8-10}
& \multicolumn{1}{c}{Men} & \multicolumn{1}{c}{Women} & \multicolumn{1}{c}{All} & \multicolumn{1}{c}{Men} & \multicolumn{1}{c}{Women} & \multicolumn{1}{c}{All} & \multicolumn{1}{c}{Men} & \multicolumn{1}{c}{Women} & \multicolumn{1}{c}{All} \\
& \multicolumn{1}{c}{(1)} & \multicolumn{1}{c}{(2)} & \multicolumn{1}{c}{(3)} & \multicolumn{1}{c}{(4)} & \multicolumn{1}{c}{(5)} & \multicolumn{1}{c}{(6)} & \multicolumn{1}{c}{(7)} & \multicolumn{1}{c}{(8)} & \multicolumn{1}{c}{(9)} \\
\midrule 
Intention & 0.378& 0.169& 0.261 & 0.366& 0.180& 0.266& & &  \\
Female & 0.000& 1.000& 0.561 & 0.000& 1.000& 0.538& 0.000 & 1.000 & 0.513 \\
Business Major & 0.126& 0.090& 0.106 & 0.127& 0.096& 0.110& 0.120 & 0.082 & 0.100 \\
STEM Major & 0.492& 0.305& 0.387 & 0.488& 0.297& 0.385& 0.464 & 0.288 & 0.374 \\
GPA & 3.427& 3.450& 3.440 & 3.413& 3.450& 3.433& 3.287 & 3.398 & 3.344 \\
ACT Math & 29.813& 27.638& 28.594 & 29.827& 27.515& 28.584& 29.393 & 27.153 & 28.245 \\
ACT Verbal & 56.956& 56.929& 56.941 & 57.349& 57.021& 57.173& 56.287 & 56.125 & 56.204 \\
Racial Minority & 0.244& 0.228& 0.235 & 0.287& 0.262& 0.274& 0.304 & 0.283 & 0.293 \\
International & 0.100& 0.065& 0.080 & 0.110& 0.071& 0.089& 0.106 & 0.072 & 0.089 \\
\midrule  
\multicolumn{1}{l}{Observations} &\multicolumn{1}{c}{10154} &\multicolumn{1}{c}{12958} &\multicolumn{1}{c}{23112} & \multicolumn{1}{c}{9147} &\multicolumn{1}{c}{10643} &\multicolumn{1}{c}{19790} &\multicolumn{1}{c}{120730} &\multicolumn{1}{c}{126956} &\multicolumn{1}{c}{247686}  \\
\bottomrule  
\end{tabular}  

}
\begin{minipage}{1\linewidth}
\smallskip
\footnotesize
\emph{Notes:} Variable means for men and women by responder group.
Early Respondents respond prior to the second data request, while Late Respondents respond any time after the second request.
\end{minipage}
\end{table}

In the raw data, men's entrepreneurial intention rate drops from 37.8\% among early respondents to 36.6\% among late respondents, while women's entrepreneurial intention rate rises from 16.9\% among early respondents to 18.0\% among late respondents.
Most of the other variables we review do not exhibit any strong pattern, with the exception of racial minority status.
We find that racial minorities are underrepresented among survey respondents, representing 23.5\% of early respondents, 27.4\% of late respondents, and 29.3\% of nonrespondents.

\subsection{Constructing a Request-Response Panel}
Our results from Section \ref{sect:Theory} involve panel data wherein subjects are observed multiple times as data requests are made, so we construct such a panel for each term.
We index terms/semesters by $s$, with retained responses, retained response choices, and accumulated requests denoted by $\hat Y_{ist}$, $\hat S_{ist}$, and $R_{ist}$, respectively.
In our setting, students who opted out of the survey in a prior term and students who respond early in a given term do not receive subsequent requests.
It follows that actual requests almost surely violate the independence assumption in Section \ref{sect:Theory}, as students receiving each number of requests differ systematically.
We address this by randomly assigning opt-outs to strata ex post and assuming $S_{ist}(r)=0$ for all $t$ and all $r$ for such students, then imputing \textit{intended} request timestamps for each student using the request timestamps of nonrespondents in their stratum.

For each student-term in a stratum, we produce observations for $t$ in $0,1,...,R_{iT}$ such that time periods are defined as intervals in which a subject has received a given number of \textit{intended} requests, with $R_{ist}=t$ for all $i$, $s$, and $t$.
Independence and monotonicity require that the times at which requests were made had representative potential responses and that subjects had sufficient time to respond between requests, respectively.
We report request timestamps in Appendix \ref{appendix:survey_timing}, revealing that the shortest interval between requests is three days, which we assume was adequate time for students to respond.
We set $S_{ist}=1$ if subject $i$'s response timestamp is between those defining the beginning and end of period $t$, with $Y_{ist}$ given by their response to the survey if $S_{ist}=1$.
We forward-impute retained response choices $\hat S_{ist}$ and retained responses $\hat Y_{ist}$ as described in Section \ref{sect:Theory} within each value of $s$.

\subsection{Nonresponse-Corrected Entrepreneurial Intentions by Gender}\label{sect:no_controls}

\begin{figure}[htbp!]
\centering
     \begin{subfigure}[t]{1\textwidth}
     \subcaption[t]{Panel A: Men}
     \centering\includegraphics[width=\linewidth]{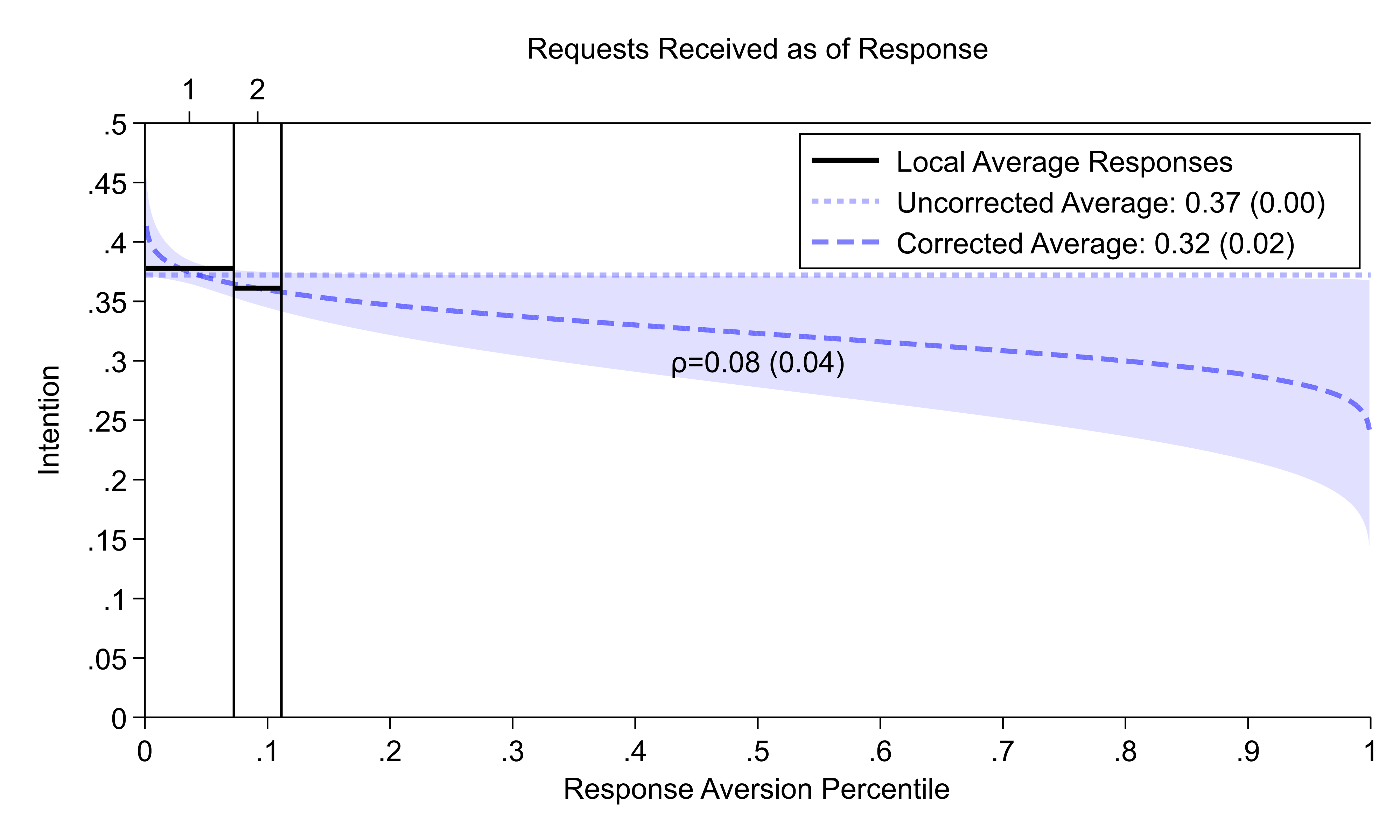}
    \end{subfigure}
    \begin{subfigure}[t]{1\textwidth}
    \subcaption[t]{Panel B: Women}
    \centering\includegraphics[width=\linewidth]{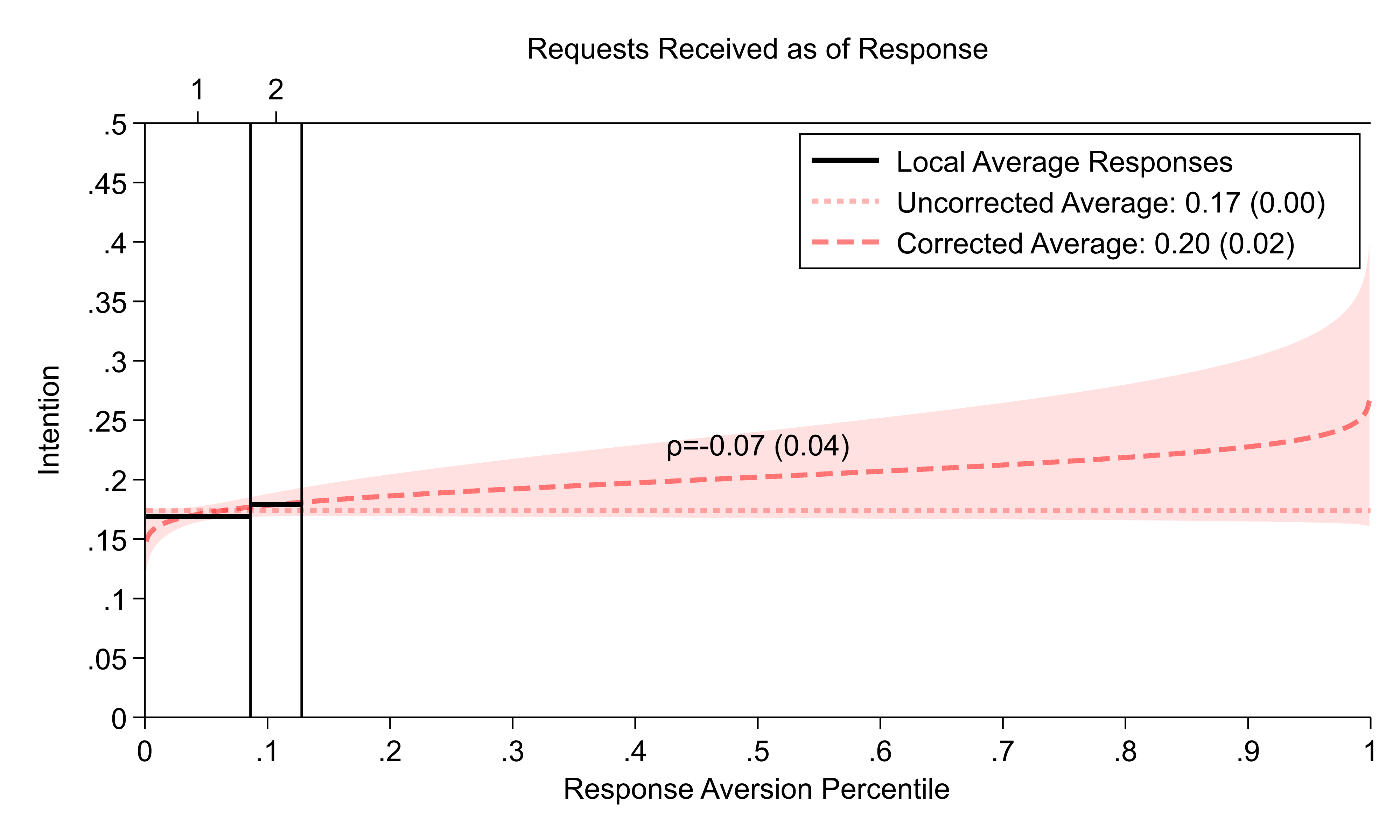}
    \end{subfigure}
    \caption{Observed and Estimated Entrepreneurial Intention, by Gender} 
    \label{fig:no_controls}
    \begin{minipage}{1\linewidth}
\smallskip
\footnotesize
\emph{Notes:} Average entrepreneurial intention rates as a function of response aversion.
Local average responses are estimates of $\beta_r$ from the equation $\hat Y_{it} = \hat S_{it}\beta_r+\epsilon_{it}$, which we estimate via 2SLS using $I(R_{it}=r)$ as an instrument for $\hat S_{it}$ in the subsample with $R_{it} \in \{r-1,r\}$, for $r=1,2$.
Corrected estimates by survey aversion percentile are obtained by estimating a \cite{van1981demand} model and calculating $\hat m(u) = \Phi((\hat\beta+\hat\rho\Phi^{-1}(1-u))/\sqrt{1-\hat\rho^2})$ for $u\in[0,1]$.
The shaded regions show 95\% confidence intervals for $ m(u)$ over the support of $u$ for both men and women.
Population average point estimates are provided in the legend, with standard errors clustered at the individual level in parentheses.
\end{minipage}
\end{figure}

We begin by estimating the raw gender gap in entrepreneurial intentions without adjusting for the covariates listed in Table \ref{tab:descriptives}.
Because total requests vary between terms in the range $[2,7]$, we drop all observations with $R_{it}>2$ to support our independence assumption for the current specification.
In addition to nonparametrically estimating local average responses for the first and second request for men and women,
we estimate gender-specific parametric \cite{van1981demand} selection models via maximum likelihood.
Specifically, we assume  \hfill

\begin{equation}\label{Binary_response_choice}
\hat S_{ist}= 
\begin{cases} 
& 1 \mbox{ if } Z_{ist}\alpha+u_{ist} \geq 0, \\
& 0 \mbox{ otherwise,}
\end{cases}
\end{equation}

\begin{equation}\label{Binary_response}
\hat Y_{ist}= 
\begin{cases} 
 1 \mbox{ if } X_{is}\beta+\epsilon_{ist} \geq 0 \quad  \& \quad & \hat S_{ist}=1 , \\
 0 \mbox{ if } X_{is}\beta+\epsilon_{ist} < 0 \quad  \& \quad & \hat S_{ist}=1, \\
 0 \mbox{ if } & \hat S_{ist}=0, 
\end{cases}
\end{equation}
and
\hfill
\begin{equation}\label{Normality}
\begin{bmatrix}
\epsilon_{ist} \\
u_{ist}
\end{bmatrix}
\sim
\mathcal{N}
\left(
\begin{bmatrix}
0 \\
0
\end{bmatrix}
,
\begin{bmatrix}
1 \qquad \rho \\
\rho \qquad 1
\end{bmatrix}
\right),
\end{equation}
where $X_{is} = 1$, and $Z_{ist} = [X_{is}, I(R_{ist}=2)]$
in this subsection for simplicity.
We assume $(\epsilon_{ist},u_{ist})$ are independent and identically distributed between $i$, with potential dependence allowed between $s$ and $t$ for each $i$.
$\rho$ gives the correlation between outcome residuals $\epsilon_{ist}$ and response preference $u_{ist} = \Phi^{-1}(1-U_{ist})$, where $U_{ist}$ is $i$'s response aversion percentile in the notation of Section \ref{sect:Theory} and $\Phi^{-1}(\cdot)$ denotes the inverse of the normal CDF.%

This gender-specific constant-only model implies that expected intention as a function of response aversion is $m(u) =\Phi((\beta+\rho\Phi^{-1}(1-u))/\sqrt{1-\rho^2})$, with the unconditional expectation  expressed as $ \int_0^1m(u)du =  \Phi(\beta)$.
We graphically present our estimates of $m(u)$ for men and women in Figure \ref{fig:no_controls} for $u \in [0,1]$.
Local average responses and sample averages for compliers to the first two requests are shown for reference.
The levels of entrepreneurial intention for early and late respondents show a decreasing pattern for men and an increasing pattern for women, as in Table \ref{tab:descriptives}.
Differences in selection bias between men and women,
reflected by their $\rho$ estimates, 
explain the difference between the 20 percentage point uncorrected gap estimate and the 12 percentage point corrected gap estimate.

Our estimates rely heavily on parametric assumptions that allow extrapolation well beyond the support of the data.
Our 15\% response rate is insufficient for informative bounds from conservative approaches such as that of \cite{horowitz2000nonparametric}, 
as can be readily seen by setting $Y_{ist}=1$ for male nonrespondents and $Y_{ist}=0$ for female nonrespondents,
and vice versa,
in Figure \ref{fig:no_controls}.
The less conservative approaches of 
\cite{lee2009training} and \cite{behaghel2015please}
would provide bounds on 
local averages of the gender gap
among students with similar levels of response aversion.
Although these bounds would 
tightly encompass the within-respondent point estimate 
in our application
due to similar response rates between men and women,
the monotonicity assumption that these methods leverage to ensure comparability between respondent subgroups is unlikely to be satisfied for gender.

\subsection{Oaxaca-Blinder-Kitagawa Decomposition}
In addition to estimating average entrepreneurial intention for men and women, we also decompose the gender intention gap using a nonlinear extension of methods developed by \cite{kitagawa1955components}, \cite{oaxaca1973male}, and \cite{blinder1973wage}.
Our decomposition estimates, for each covariate in $X_i$, the extent to which the gender gap would be reduced if women's average value of that single variable were equal to that of men, and how much it would be reduced if that variable's coefficient for women were equal to the coefficient for men. These results are shown in Table \ref{tab:Oaxaca_intention}.

To implement the decomposition, we estimate the model in (\ref{Binary_response_choice})-(\ref{Normality}) jointly for men and women via maximum likelihood using gender-specific parameters $(\beta_m,\alpha_m,\rho_m)$ for men and $(\beta_w,\alpha_w,\rho_w)$ for women.
To avoid overweighting terms in which more requests were sent, 
we weight each observation by the inverse of the total number of requests sent to that student in that term.
We define $X_{is}$ as all covariates listed in Table \ref{tab:descriptives} along with term fixed effects, and we define $Z_{ist}$ as binary indicators for each value in $R_{ist}=1,2...,7$ interacted with $X_{is}$.%
\footnote{
\cite{blandhol2022tsls} recommend rich covariate interactions for instrumental variables when estimating local average treatment effects via two stage least squares.
The estimator and target parameters are different here, but we find their argument convincing that if model assumptions relating to instrumental variables hold conditional on $X_i$, then $X_i$ should be flexibly interacted with instruments to support identification.
}
We order the sample of $N_w$ women and $N_m$ men such that subjects with $i=1,...,N_w$ are women and those with $i = N_w+1,...,N$ are men, with $N= N_w+N_m$.

\begin{table}[htbp!]\centering
\def\sym#1{\ifmmode^{#1}\else\(^{#1}\)\fi}
\caption{Oaxaca-Blinder-Kitagawa Decomposition of the Entrepreneurial Intention Gender Gap}
\scriptsize{
\label{tab:Oaxaca_intention}
\begin{tabular}{@{\extracolsep{4pt}}l*{7}{c}}
\hline
& \multicolumn{3}{c}{Uncorrected} & \multicolumn{3}{c}{Corrected} & \\
\cline{2-4} \cline{5-7} 
& Men & Women & Gap & Men & Women & Gap  \\
Variables & (1) & (2) & (3) & (4) & (5) & (6) \\
\hline
Intention & 0.381 & 0.179 & 0.202  & 0.375 & 0.192 & 0.183  \\
 & (0.005)  & (0.003)  & (0.005)  & (0.020)  & (0.015)  & (0.025) \\
[1em]
X \\
Business Major  & 0.120 & 0.084  & 0.004  & 0.120 & 0.084  & 0.004 \\
 & (0.001) & (0.001) & (0.000) & (0.001) & (0.001) & (0.000) \\
STEM Major  & 0.468 & 0.290  & -0.005  & 0.468 & 0.290  & -0.006 \\
 & (0.001) & (0.001) & (0.001) & (0.001) & (0.001) & (0.001) \\
GPA  & 3.305 & 3.406  & 0.005  & 3.305 & 3.406  & 0.005 \\
 & (0.001) & (0.001) & (0.001) & (0.001) & (0.001) & (0.001) \\
ACT Math  & 29.452 & 27.221  & 0.003  & 29.452 & 27.221  & 0.003 \\
 & (0.011) & (0.010) & (0.002) & (0.011) & (0.010) & (0.002) \\
ACT Verbal  & 56.405 & 56.258  & -0.000  & 56.405 & 56.258  & -0.001 \\
 & (0.023) & (0.022) & (0.000) & (0.023) & (0.022) & (0.000) \\
Racial Minority  & 0.299 & 0.277  & 0.001  & 0.299 & 0.277  & 0.001 \\
 & (0.001) & (0.001) & (0.000) & (0.001) & (0.001) & (0.000) \\
International  & 0.105 & 0.072  & 0.005  & 0.105 & 0.072  & 0.005 \\
 & (0.001) & (0.001) & (0.000) & (0.001) & (0.001) & (0.001) \\
All X & & & 0.013 & & & 0.013 \\
 & & & (0.002) & & & (0.003) \\
[1em]
$\beta$ \\
Business Major  & 0.449 & 0.457  & -0.000  & 0.439 & 0.457  & -0.000  \\
 & (0.036) & (0.038) & (0.001) & (0.037) & (0.039) & (0.002) \\
STEM Major  & -0.039 & -0.118  & 0.005  & -0.041 & -0.122  & 0.006  \\
 & (0.024) & (0.027) & (0.003) & (0.025) & (0.028) & (0.003) \\
GPA  & -0.181 & -0.198  & 0.015  & -0.186 & -0.198  & 0.010  \\
 & (0.024) & (0.026) & (0.031) & (0.026) & (0.027) & (0.034) \\
ACT Math  & 0.007 & 0.006  & 0.012  & 0.008 & 0.006  & 0.014  \\
 & (0.004) & (0.004) & (0.039) & (0.004) & (0.004) & (0.042) \\
ACT Verbal  & -0.013 & -0.013  & 0.008  & -0.013 & -0.013  & 0.010  \\
 & (0.002) & (0.002) & (0.034) & (0.002) & (0.002) & (0.037) \\
Racial Minority  & 0.122 & 0.191  & -0.006  & 0.105 & 0.182  & -0.007  \\
 & (0.032) & (0.030) & (0.004) & (0.034) & (0.031) & (0.004) \\
International  & 0.379 & 0.590  & -0.006  & 0.403 & 0.609  & -0.006  \\
 & (0.049) & (0.053) & (0.002) & (0.051) & (0.055) & (0.002) \\
All $\beta$ &  &  & 0.030 &  &  & 0.030 \\
 &  &  & (0.042) &  &  & (0.047) \\
[1em]
All Unexplained Gaps & & & 0.159 & & & 0.140 \\
 & & & (0.043) & & & (0.059) \\
[1em]
Auxiliary Parameters \\
$\rho$ & 0 & 0 & 0 & 0.010 & -0.032 & 0 \\
& (.) & (.) & (.) & (0.034) & (0.035) & (.) \\
[1em]
Log-Likelihood & -67,349 & -74,546 & -141,895 & -57,558 & -64,508 & -122,067 \\
Observations  & 140,031 & 150,557 & 290,588 & 604,053 & 652,310 & 1,256,363 \\
Responses  & 19,301 & 23,601 & 42,902 & 65,959 & 79,549 & 145,508 \\
\hline
Model Tests & Men & Women & Joint & Gap \\ 
p-value: $\mathbb{E}[\hat{\mathbb E}[Y_i^*|S_i=1,X_i]] = \mathbb{E}[\hat{\mathbb E}[Y_i^*|X_i]]$ & 0.740 & 0.398 & 0.663 & 0.447  \\
p-value: $\rho=0$ & 0.759 & 0.362 & 0.629 & 0.385 \\
p-value: $\beta_R=0$ (no request effects) & 0.135 & 0.411 & 0.200 \\
\hline
\end{tabular}

}
\begin{minipage}{1\linewidth}
\smallskip
\footnotesize
\emph{Notes:}
Uncorrected estimates obtained by estimating a \cite{van1981demand} model using only the final request observation for each student-term with the constraint $\rho=0$.
Corrected estimates obtained by estimating the same model using all student-term-request observations with unconstrained $\rho$, with observations weighted by the inverse of total requests by student-term.
Group-specific population mean estimates of variables and parameters, as indicated, are in columns (1), (2), (4), and (5), with sample means of $X$ calculated across student-term observations.
Columns (3) and (6) give uncorrected and corrected estimated contributions of each variable or coefficient to the estimated gap in entrepreneurial intentions, using equations (\ref{x_gap}), (\ref{b_gap}), (\ref{X_gap}), (\ref{B_gap}), and (\ref{Unexplained_gap}).
$\mathbb{E}[\hat{\mathbb E}[Y_{is}^*|S_{isT}=1,X_i]]$ is estimated as the average of $\Phi(X_{is} \hat \beta^{uncorrected})$ and $\mathbb{E}[\hat{\mathbb E}[Y_{is}^*|X_{is}]]$ is estimated as the average of $\Phi(X_{is} \hat \beta^{corrected})$ over student-term observations.
Term fixed effects, included in $X_i$ to support the independence assumption, are treated as part of the unexplained gap.
Standard errors clustered at the individual level are in parentheses.
\end{minipage}
\end{table}

We estimate the contribution of $x_{is} \subset X_{is}$ to the entrepreneurial intention gap as
\begin{equation}\label{x_gap}
    \hat\Delta_x = 
    \sum_{i=1}^{N_w}
    \frac{
    \Phi(X_{is}\hat\beta_w+(\bar x_{m}-\bar x_{w})\hat\beta_{w,x})
    -\Phi(X_{is}\hat\beta_w)
    }
    {N_w
    },
\end{equation}
wherein we denote the sample average of $x_i$ for men and women as $\bar x_{m}$ and $\bar x_{w}$, respectively, and we 
denote male and female coefficients on $x_i$ as $\beta_{m,x}$ and $\beta_{w,x}$, respectively.%
\footnote{Fairlie's \citeyearpar{fairlie2005extension} similar nonlinear decomposition considers the case in which the marginal distribution of $x_{is}$ for women is set equal to men's, while ours considers the case in which the mean of $x_{is}$ for women is set equal to men's via a level shift of women's distribution of $x_{is}$.
We view the counterfactual we consider as likely simpler to understand and simpler to implement for policymakers, though describing $\beta_{w,x}$ as the ``effect'' of changing $x_{is}$ for women requires assumptions beyond those that either we or \cite{fairlie2005extension} make.
}
We similarly estimate the contribution to the average gender gap of $x_{is}$'s differential effect on $Y^*_{is}$ between men and women as \hfill
\begin{equation}\label{b_gap}
    \hat\Delta_{\beta_x} = 
    \sum_{i=1}^{N_w}
    \frac{
    \Phi(X_{is}\hat\beta_w+x_{is}(\hat\beta_{m,x}-\hat\beta_{w,x}))
    -\Phi(X_{is}\hat\beta_w)
    }
    {N_w
    }.
\end{equation}
In addition to estimating the contribution of each variable and its coefficient on the gender-intention gap, we also estimate the contribution of the entire vector $X_{is}$ on the gap as
\begin{equation}\label{X_gap}
    \hat\Delta_X = 
    \sum_{i=1}^{N_w}
    \frac{
    \Phi(X_{is}\hat\beta_w+(\bar X_{m}-\bar X_{w})\hat\beta_{w})
    -\Phi(X_{is}\hat\beta_w)
    }
    {N_w
    },
\end{equation}
defining the vector of sample means of all covariates in $X_{is}$ for men as $\bar X_{m}$ and for all women as $\bar X_{w}$.
Similarly, we estimate the contribution of the entire vector of gender differences in coefficients on the gender-intention gap as
\hfill
\begin{equation}\label{B_gap}
    \hat\Delta_{\beta} = 
    \sum_{i=1}^{N_w}
    \frac{
    \Phi(X_{is}\hat\beta_w+X_{is}(\hat\beta_{m}-\hat\beta_{w}))
    -\Phi(X_{is}\hat\beta_w)
    }
    {N_w
    }.
\end{equation}
We estimate the remaining gap, which is driven by the interaction of male-female differences in $X$ and male-female differences in $\beta$, distributional differences in $X$ between men and women, and nonlinearity of the normal distribution as
\hfill
\begin{equation}\label{Unexplained_gap}
    \hat\Delta_{R} = 
    \sum_{i=N_w+1}^{N}
    \frac{
    \Phi(X_{is}\hat\beta_m)
    }
    {N_m}
    -
    \sum_{i=1}^{N_w}
    \frac{
    \Phi(X_{is}\hat\beta_w)
    }
    {N_w}
    -
    \hat\Delta_{X}
    -
    \hat\Delta_{\beta}
    .
\end{equation}

Table \ref{tab:Oaxaca_intention} contains estimated means of the variable or coefficient listed in each row for men, women, and the contribution of that item to the intention gap, using predicted values from a naive probit in columns 1–3, and using our proposed selection correction model in columns 4–6.
$X_{is}$ is observed for all individuals regardless of survey response, so the means of $X_{is}$ by gender are the same for the uncorrected model and the corrected model.
Columns 3 and 6 report estimated contributions of each row quantity to the gender gap as described in equations (\ref{x_gap}), (\ref{b_gap}), (\ref{X_gap}), (\ref{B_gap}), and (\ref{Unexplained_gap}).
We perform our decomposition for all variables in Table \ref{tab:descriptives}, with term fixed effects included in $X_{is}$ contributing only to the unexplained gap.

We estimate the overall gender gap by calculating
\begin{equation}\label{Total_gap}
    \hat\Delta = 
    \sum_{i=N_w+1}^{N}
    \frac{
    \Phi(X_{is}\hat\beta_m)
    }
    {N_m}
    -
    \sum_{i=1}^{N_w}
    \frac{
    \Phi(X_{is}\hat\beta_w)
    }
    {N_w}
    ,
\end{equation}
where our estimates of $(\hat \beta_w, \hat \beta_m)$ differ by estimation method.
The estimated gender-intention gap from our preferred specification is 18.3 percentage points, with a naive probit producing an estimate of 20.2 percentage points.
Differences between the corrected and uncorrected model are due to 
the use of only observations with $t=T$ for the naive probit
and the inclusion of the $\rho$ parameter in the corrected model (with $\rho=0$ fixed in the uncorrected model) which is positive if early respondents have higher intention and negative if they have lower intention.
Our estimates of $\rho$ in the current specification are insignificant, suggesting limited nonresponse bias, though we see the same pattern with the signs of our point estimates here as we did in the no-controls specification shown in Figure \ref{fig:no_controls}.
Along these same lines, our richer specification repeats the finding in Figure \ref{fig:no_controls} that the uncorrected gap point estimate is larger than the corrected estimate, but here we fail to reject that they are the same in the population (p=0.449).

The difference in corrected estimates in Section \ref{sect:no_controls} and the corrected estimates in Table \ref{tab:Oaxaca_intention} is due to the inclusion of controls and observations for requests three through seven, which were observed in some years but not others.
The richer nonresponse correction specification in this section essentially fits a curve 
through local responses as in Figure \ref{fig:no_controls} for each value of $X_{is}$, 
where the slope of the curve is determined by a gender-specific $\rho$ that is held fixed across $X_i$.
For both men and women, $\rho$ is attenuated in the richer model, suggesting that observed covariates account for both late responses and differences in intention between late and early responses.
Nonetheless, the same general pattern from Figure \ref{fig:no_controls} persists, with positive selection for men ($\hat\rho_m>0)$ and negative selection for women $(\hat \rho_w<0)$, though these estimates are not statistically significant in this specification.

Our estimates are only reliable if the assumptions in Section \ref{sect:Theory} and equations (\ref{Binary_response_choice})-(\ref{Normality}) hold for this application.
Under these assumptions, all local average responses should be explained by the marginal survey response function.
We assess this by performing a ``predictable trends'' overidentification test wherein we include indicators for requests $R_{it}=3,...,7$ in $X_i$ with coefficient vector $\beta_R$, 
effectively relying on only requests one and two for identification of the $\rho$ parameters.
We fail to reject the null that later requests have no effect on responses at conventional levels for men (p=0.135), for women (p=0.411), and for both jointly (p=0.200).
However, these point estimates are narrowly insignificant and mostly positive for men as shown in Figure \ref{fig:overid}, suggesting that later requests drive corrected point estimates upward toward the uncorrected estimates for men when they are included, increasing the estimated gender gap.

\begin{figure}[htbp!]
\centering
        \includegraphics[width=\linewidth]{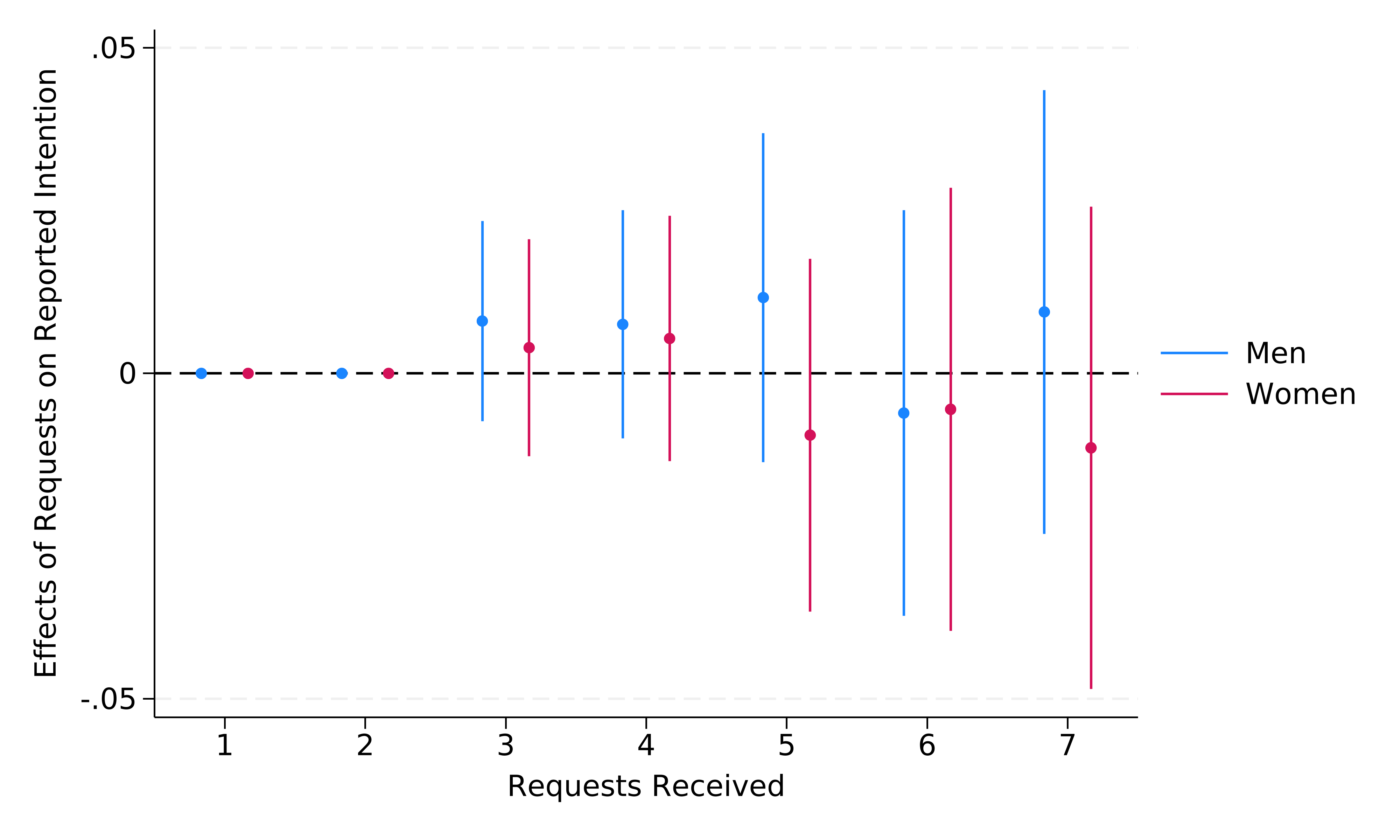}
    \caption{Event Study of Group Means Relative to Model Predictions, by Gender} \label{fig:overid}
    \begin{minipage}{1\linewidth}
\smallskip
\footnotesize
\emph{Notes:} 
We depict the probit coefficient point estimates and 95\% confidence intervals for gender-specific estimated effects of requests $R_{it}>2$ on intention.
We fail to reject that all effects are zero, with p=0.135 for men, p=0.411 for women, and p=0.200 for men and women jointly.
\end{minipage}
\end{figure}

Most of the entrepreneurial intention gender gap is unexplained by our observed variables.
While the contribution of differences in observed covariates to the gender gap of 1.3 percentage points is statistically significant at conventional levels, it explains a small fraction of the overall gap.
Meanwhile, the total corrected contribution of gender differences in $\beta$ to the gender gap is larger at 3 percentage points, but it is not significant at conventional levels.
The largest statistically significant individual contributors to the gender intention gap are gender differences in the rates of choosing a STEM major and gender differences in the relative rates of intention between STEM and non-STEM students.
Interestingly, these factors approximately cancel each other out.
On one hand, there are more male STEM majors, who have somewhat lower entrepreneurial intention than male non-STEM majors.
Meanwhile, there are fewer female STEM majors, but female STEM majors have much lower entrepreneurial intention than female non-STEM majors.

Results from our preferred specification suggest a large gender gap in undergraduate entrepreneurial intention, with  minimal nonresponse bias in this survey.
Despite the apparent lack of significant nonresponse bias, our corrected estimates support substantially different inferences than the uncorrected estimates because of their larger standard errors.
Perhaps the clearest evidence of this is that the uncorrected gap estimate strongly rejects a true population average equal to the corrected point estimate.
Given that our parametric correction gives substantially more precise estimates than conservative approaches such as that of \cite{horowitz1998censoring}, we recommend against drawing inferences while assuming no nonresponse bias in applications such as ours.

\section{Conclusion}\label{sect:Conclusion}
We have established that data requests made to subjects are valid nonresponse correction instruments in panel data containing their accumulated requests and responses over time, 
under assumptions similar to those routinely made when analyzing survey data
that closely correspond to those that enable treatment effect identification using before-after comparison designs.
Specifically, we assume that potential responses are representative regardless of response timing, 
requests should not directly affect responses, 
and subjects do not systematically respond at times when their potential responses are inaccurate.
Despite the strength of these 
assumptions,
which has led to a large literature devoted to avoiding them in causal inference applications such as the seminal work by \cite{card1994minimum},
research relying on survey responses remains popular.
In addition to these assumptions, we make a monotonicity assumption similar to that of \cite{imbens1994identification} which allows us to order subjects by their response timing and extrapolate from respondents to nonrespondents using existing nonresponse correction methods.

Our procedure for constructing request instruments imposes minimal challenges for data collection and analysis.
It requires no knowledge of econometrics by data collection administrators and no financial incentives for subjects, which can bias responses if subjects maximize their effective wages by minimizing time spent on the survey, such as by using artificial intelligence to concoct responses \citep{westwood2025}.
Meanwhile, it naturally reveals rich variation in response rates when multiple reminders are sent, enabling model validation tests that compare late responses to their predicted values implied by early responses.
These features enable routine use of our method when analyzing data voluntarily submitted by subjects in response to requests, 
either via surveys or other voluntary data submissions
such as lab testing in longitudinal medical experiments.

We illustrate our method by estimating the gender gap in entrepreneurial intentions among undergraduate students at the University of Wisconsin-Madison.
To do so, we construct a panel dataset that contains a row for each student for each data request they receive, with their accumulated requests recorded in all rows and their responses recorded only in rows representing time periods after they respond.
We estimate a large gender gap,
with
uncorrected estimates strongly rejecting the point estimate of the average gender gap obtained when correcting for nonresponse.
In an additional application in Supplemental Appendix \ref{appendix:selection_in_surveys}, we accurately and precisely estimate ground truth averages from administrative data for five out of six labor market indicators using only aggregated statistics from early and late responses to the 45\% response-rate Norway in Corona Times survey.
Our empirical results suggest that correcting for nonresponse using nonrandom request instruments in subject contact panel data can improve inferences, while posing relatively small challenges for data collection and analysis.

\clearpage

\bibliographystyle{econometrica.bst}
\renewcommand{\bibname}{BIBLIOGRAPHY}
\bibliography{References}


\appendixpageoff
\begin{appendices}
\numberwithin{equation}{section}
\makeatletter 
\newcommand{\section@cntformat}{Appendix \thesection:\ }
\makeatother

\renewcommand{\thesubsection}{\Alph{section}.\arabic{subsection}}

\setcounter{table}{0}
\renewcommand{\thetable}{\Alph{section}.\arabic{table}}
\setcounter{figure}{0}
\renewcommand{\thefigure}{\Alph{section}.\arabic{figure}}

\section{Survey Timing}\label{appendix:survey_timing}
Here we present additional details regarding the implementation of the student entrepreneurship survey that we analyze in Section \ref{sect:application}.
In each term, survey administrators randomly assigned students into approximately ten strata using a random number generator in Stata.
The stratification allowed survey administrators to confirm that request emails were error free for each group prior to sending them to the next group.
This feature of the survey administration caused the time-stamps for requests to vary between strata within each term.

We report the modal dates of each request in each term alongside their cumulative response rates in Table \ref{tab:survey_timing}.
We observe a secular trend in the number of requests over time, with a downward trend in responsiveness to earlier requests.
These apparent systematic differences between terms motivate our use of term fixed effect controls and term$\times$request interactions for our preferred specification in Table \ref{tab:Oaxaca_intention}, which approximately averages results from separate term-specific models.

Independence and monotonicity require that responses received in the intervals between time periods are representative and that subjects have sufficient time to respond to each request prior to receiving the next, respectively.
Independence could be violated if particular events had systematic effects on self-reported entrepreneurship, which we do not expect for entrepreneurial intention.%
Monotonicity could be violated for some of our requests if some subjects went several days without checking their emails, though these violations would be inconsequential if slow responses indicate relatively high survey aversion among compliers to each request.%
\footnote{
Formalizing the relationship between response speed for each request and response aversion exceeds the scope of this paper, but
if elapsed time since the first request were a valid instrument used in place of requests, 
late compliers to request $r$ would have survey aversion between early $r$ compliers and early $r+1$ compliers.
}

\begin{table}[htbp!]\centering
\def\sym#1{\ifmmode^{#1}\else\(^{#1}\)\fi}
\caption{Cumulative Response Rates by Term and Request}
\scriptsize{
\label{tab:survey_timing}
\begin{tabular}{@{\extracolsep{4pt}}l*{9}{c}}
Term  & Request 1& Request 2& Request 3& Request 4& Request 5& Request 6& Request 7\\
\midrule 
2015, Fall \\
Request Date & 11/23 & 12/3 &  &  &  &  &   \\
Cum. Response Rate &  0.16 &  0.19 & & & & &  \\
[1em] 
2016, Fall \\
Request Date & 10/24 & 12/12 & 12/21 &  &  &  &   \\
Cum. Response Rate &  0.11 &  0.17 &  0.19 & & & &  \\
[1em] 
2017, Fall \\
Request Date & 11/7 & 11/29 & 12/5 &  &  &  &   \\
Cum. Response Rate &  0.13 &  0.17 &  0.19 & & & &  \\
[1em] 
2018, Fall \\
Request Date & 10/30 & 11/8 & 11/29 & 12/15 & 12/26 &  &   \\
Cum. Response Rate &  0.07 &  0.11 &  0.12 &  0.14 &  0.15 & &  \\
[1em] 
2019, Fall \\
Request Date & 11/1 & 11/7 & 12/9 & 12/21 & 12/29 &  &   \\
Cum. Response Rate &  0.07 &  0.11 &  0.12 &  0.12 &  0.13 & &  \\
[1em] 
2020, Spring \\
Request Date & 4/22 & 4/27 & 4/30 & 5/11 &  &  &   \\
Cum. Response Rate &  0.06 &  0.11 &  0.14 &  0.15 & & &  \\
[1em] 
2020, Fall \\
Request Date & 10/8 & 10/15 & 10/20 & 11/2 & 11/17 & 11/30 &   \\
Cum. Response Rate &  0.08 &  0.12 &  0.14 &  0.16 &  0.17 &  0.19 &  \\
[1em] 
2021, Spring \\
Request Date & 4/21 & 4/24 & 4/30 & 5/11 &  &  &   \\
Cum. Response Rate &  0.06 &  0.10 &  0.11 &  0.12 & & &  \\
[1em] 
2021, Fall \\
Request Date & 10/20 & 10/25 & 11/19 & 11/30 & 12/9 & 12/13 & 12/28  \\
Cum. Response Rate &  0.07 &  0.11 &  0.12 &  0.13 &  0.14 &  0.16 &  0.16  \\
[1em] 
2022, Fall \\
Request Date & 10/27 & 11/1 & 11/17 & 11/29 & 1/4 &  &   \\
Cum. Response Rate &  0.04 &  0.07 &  0.09 &  0.09 &  0.10 & &  \\
[1em] 
\bottomrule 
\end{tabular}  

}
\begin{minipage}{1\linewidth}
\smallskip
\footnotesize
\emph{Notes:}
Request timing and cumulative response rates by term.
\end{minipage}
\end{table}

\section{Revisiting the Norway in Corona Times Survey}\label{appendix:selection_in_surveys}
In this section, we demonstrate the performance of our method on a synthetic version of the Norway in Corona Times (NCT) survey, using only estimated quantities reported by \cite{dutz2021selection}, henceforth referred to as DHLMTV.%
\footnote{
We use the version revised February 2025, which we downloaded from NBER on February 24, 2025.
}
This survey requested information from 10,000 randomly selected adult Norwegians regarding their employment and wages before and after COVID-19 lockdowns and obtained approximately a 50\% response rate.
DHLMTV merge the survey with administrative data on survey respondents to investigate monthly earnings before lockdown (converted to USD), monthly earnings after lockdown, the share of individuals with an earnings loss above 20\%, the share of individuals employed before lockdown, the share of individuals employed after lockdown, and the share of individuals who experienced an employment loss after the lockdown.
Importantly for us, these authors provide estimated averages of variables for early and late responders to the NCT survey as well as ground truth values taken from Norwegian administrative data on the same variables.

We begin by matching the moments from DHLMTV's Table 4.
DHLMTV focus on the 93\% of subjects who were invited to take the survey online.
We further restrict our focus to the 40\% of subjects who received no randomized survey incentives,
so as 
to illustrate our method's performance using only requests as instruments.
We therefore construct a synthetic sample of $3,720 = 10,000 \times 0.93 \times 0.40$ subjects.
Following their description of their methods, we assume that the always taker estimates they report in their Table 4 are consistent estimates for average early responses, and that the reminder complier estimates in their Table 4 are consistent estimates for average late responses.

We produce synthetic data to match DHLMTV Table 4, defining the first 38\% of subjects as always-takers, the next 7\% as reminder compliers, and the last 55\% as nonrespondents.
For continuous variables, we first set the means for each group equal to those reported, leaving nonrespondents missing.
Next, we then add the standard deviations, backed out from the reported group-specific standard errors, to half of the subjects in each group and subtract them from the other half.
For binary variables, we set variable values for always takers and reminder compliers to 1 and 0 as needed to make the group-specific averages match those reported in DHLMTV Table 4.
We set variables to missing for the remaining 55\% of the observations that do not correspond to always-takers or reminder compliers.
We then reshape the data so that each subject has one observation each for $R_{i0}=0$, $R_{i1}=1$, and $R_{i2}=2$, where we initialize $\hat S_{it} = 0$ for all $i$ and $t$, replaced with $\hat S_{i1} = 1$ for always-takers and $\hat S_{i2} = 1$ for both always-takers and reminder compliers.

\begin{table}[htbp!]\centering
\def\sym#1{\ifmmode^{#1}\else\(^{#1}\)\fi}
\caption{Methods Comparison using Synthetic Norway in Corona Times Survey}
\label{tab:selection_in_surveys}
\begin{tabular}{@{\extracolsep{4pt}}l*{7}{c}}
\hline
& \multicolumn{3}{c}{Earnings} & \multicolumn{3}{c}{Employment} & \\
\cline{2-4} \cline{5-7} 
& Before & After & Large Loss & Before & After & Loss  \\
Estimator & (1) & (2) & (3) & (4) & (5) & (6) \\
\hline
Panel 1: Ground truth & 3,095 & 2,981 & 0.148 & 0.567 & 0.494 & 0.091 \\
\\
[1em]
Panel 2: Sample Means \\
Always Taker (38\%) & 3,746 & 3,783 & 0.13  & 0.65 & 0.64 & 0.03  \\
 & (116)  & (107)  & (0.01)  & (0.01)  & (0.01)  & (0.00) \\
Reminder Complier (7\%) & 3,244 & 3,257 & 0.12  & 0.55 & 0.55 & 0.03  \\
 & (256)  & (251)  & (0.02)  & (0.03)  & (0.03)  & (0.01) \\
Always+Reminders (45\%) &     3,668 &     3,701 &  0.13 &  0.63 &  0.63 &  0.03  \\ 
& (106) & (98) & (0.00) & (0.01) & (0.01) & (0.00)  \\ 
 \\
[1em]
Panel 3: DHLMTV \\
(Mid)point & 3,368 & 3,232 & 0.142 & 0.588 & 0.536 & 0.091\\
Bounds & & & & [0.567, 0.609] \\
[1em]
Panel 4: Our Method \\
MLE &     3,197 &     3,211 & 0.116 & 0.517 & 0.521 & 0.026  \\ 
& (250) & (238) & (0.023) & (0.039) & (0.039) & (0.011)  \\ 
Two step &     3,107 &     3,113 & 0.116 & 0.520 & 0.523 & 0.026  \\ 
& (303) & (297) & (0.026) & (0.037) & (0.037) & (0.014)  \\ 
 \\
\hline
Observations & 3,720 & 3,720 & 3,720 & 3,720 & 3,720 & 3,720 \\
\hline
\end{tabular}
\begin{minipage}{1\linewidth}
\smallskip
\footnotesize
\emph{Notes:}
Estimates of earnings and employment before and after COVID-19 lockdowns using the NCT survey.
Panels 1-3 reconstructed from selected rows of Tables 4 and 5 of Dutz et al (2025). 
Row 1 of Panel 4 reports full information maximum likelihood estimates of \cite{heckman1979sample} and \cite{van1981demand} selection models for continuous and binary variables, respectively, with standard errors clustered at the individual level in parenthesis.
Row 2 of Panel 4 reports estimates from the two-step estimator of \cite{heckman1979sample}, with bootstrapped standard errors with individual-clustered resampling in parentheses.
\end{minipage}
\end{table}

Our Table \ref{tab:selection_in_surveys} compares estimates using our method with those of Tables 4 and 5 of DHLMTV.
Our results in Section \ref{sect:Theory} describe nonparametric identification of
local average responses
using request instruments, 
but do not stipulate any particular means of extrapolation to population averages.
For simplicity, we present point estimates obtained from the full information maximum likelihood version of the \cite{heckman1979sample} selection correction for continuous variables as well as the probit version described by \cite{van1981demand} for binary variables.
We also present estimates from the two-step correction of \cite{heckman1979sample} for both continuous and binary variables.

We compare estimates of survey variable averages using request instruments to ground truth averages, the pooled means of always-takers and reminder compliers, and DHLMTV's estimates from their specification reported in their Table 5.
DHLMTV's specification makes use of multiple requests, randomized incentives, and covariate adjustment in the context of a two-dimensional response aversion model.
We refer readers to DHLMTV for further details on their assumptions and estimation procedures.
To reiterate, our method estimates a one-dimensional response aversion model using multiple requests without using randomized incentives or covariate adjustment.
Our method is capable of making use of covariates, but we do not have access to such data for the NCT survey, so we settle for only making use of the publically reported estimated means for always takers and reminder compliers for each outcome variable.

Our assessment is that our method performs better than the uncorrected means denoted in Table \ref{tab:selection_in_surveys} as ``Always+Reminder.''
Specifically, our point estimates are further from the ground truth than the uncorrected estimates for only the Large Loss variable, they are equally distant for the Loss variable, and they are closer for all others.
Meanwhile, we obtain larger standard errors which allow us to better avoid the illusion of certainty in cases where point estimates are imperfect.
For instance, although our MLE population average point estimate for Large Loss is further from the ground truth than the sample mean, 
the ground truth is covered by the corrected mean's 95\% confidence interval 
and 
is not covered by the uncorrected mean's 95\% confidence interval.
This pattern is repeated for other outcomes, where our preferred method incorrectly rejects the ground truth with 95\% confidence only for employment loss, while sample means differ from ground truth averages of all variables with 95\% confidence.

Relative to the method of DHLMTV, our method performs somewhat better with continuous variables and substantially worse with binary variables.
Specifically, our point estimates are closer to the ground truth than those of DHLMTV for earnings before, earnings after, and employment after, where as their method's estimates are closer to the ground truth than ours for large earnings losses, employment before, and employment loss.
Our method performs particularly poorly with employment loss, which exhibits no variation between always takers and reminder compliers, leading our method to reproduce the in-sample mean as its best guess at the population mean.
Meanwhile, DHLMTV's method obtains the ground truth for this variable exactly.

Our method performs at its worst for estimating a large loss in earnings and employment loss.
One explanation for this is that these variables have means near zero, and our normality assumptions may be particularly at odds with reality for such variables.
Another explanation is that individuals who have experienced negative outcomes may be particularly resistant to responding to (poorly incentivized) surveys about those outcomes. 
We are unable to distinguish these explanations with the data at hand, so we end this section by recommending particular caution to researchers using parametric extrapolations estimated via request instruments for rare outcomes and/or those that represent negative life outcomes.
Our conclusion from this section is that methods that leverage both multiple requests and randomized incentives, such as that of DHLMTV, may perform better than only using multiple requests in some cases, 
but that our method is likely preferable to uncorrected estimates in cases where randomized survey response incentives were not administered.

\end{appendices}

\end{document}